\documentclass[%
 reprint,
 onecolumn,
 amsmath,
 amssymb,
 aps,
 unsortedaddress
]{revtex4-2}

\usepackage{graphicx}
\usepackage{subcaption}
\usepackage{amsmath,amssymb,amsthm,easybmat,verbatim}
\usepackage[margin=1in]{geometry}

\usepackage{caption}
\usepackage{xcolor}
\usepackage{listings}
\usepackage{hyperref}
\usepackage{float}
\hypersetup{
    colorlinks,
    linkcolor=blue,
    citecolor=blue,
    urlcolor=blue
}

 \newcommand{\ket}[1]{{\left| #1 \right\rangle}}
\usepackage{orcidlink}

\usepackage{colortbl}
\usepackage{tikz}
\usetikzlibrary{positioning,fit,backgrounds, calc}

\newtheorem{theorem}{Theorem}[section]
\newtheorem{proposition}[theorem]{Proposition}
\newtheorem{lemma}[theorem]{Lemma}
\newtheorem{corollary}[theorem]{Corollary}
\theoremstyle{definition}
\newtheorem{definition}[theorem]{Definition}

\newcommand{\halfedge}[3]{%
  \coordinate (he@m) at ($(#1)!0.5!(#2)$);
  \pgfmathanglebetweenpoints{\pgfpointanchor{#1}{center}}{\pgfpointanchor{#2}{center}}%
  \edef\he@ang{\pgfmathresult}%
  \begin{scope}
    \path[clip,overlay,shift={(he@m)},rotate=\he@ang] (-30,-30) rectangle (0,30);
    \draw[red,line width=0.9pt] (#1) to[bend left=#3] (#2);
  \end{scope}
  \begin{scope}
    \path[clip,overlay,shift={(he@m)},rotate=\he@ang] (0,-30) rectangle (30,30);
    \draw[blue,line width=0.9pt] (#1) to[bend left=#3] (#2);
  \end{scope}%
}

\newcommand{\halfedgehi}[3]{%
  \coordinate (he@m) at ($(#1)!0.5!(#2)$);
  \pgfmathanglebetweenpoints{\pgfpointanchor{#1}{center}}{\pgfpointanchor{#2}{center}}%
  \edef\he@ang{\pgfmathresult}%
  \begin{scope}
    \path[clip,overlay,shift={(he@m)},rotate=\he@ang] (-30,-30) rectangle (0,30);
    \draw[red,line width=2.6pt,opacity=0.95] (#1) to[bend left=#3] (#2);
  \end{scope}
  \begin{scope}
    \path[clip,overlay,shift={(he@m)},rotate=\he@ang] (0,-30) rectangle (30,30);
    \draw[blue,line width=2.6pt,opacity=0.95] (#1) to[bend left=#3] (#2);
  \end{scope}%
}

\tikzset{
  vtx/.style={circle,draw=black,fill=white,line width=0.8pt,
              minimum size=7mm,inner sep=0pt,font=\small},
  apex/.style={vtx,fill=black!8},
  spoke/.style={blue,line width=0.9pt},
  spokehi/.style={blue,line width=2.6pt},
}

\newcommand{\placevertices}[3]{%
  \foreach \i in {1,...,#1}{%
    \pgfmathsetmacro{\ang}{90+(\i-1)*360/#1}
    \coordinate (v\i) at (\ang:#2);
  }
  \coordinate (va) at (0,0);
}
\newcommand{\drawvertices}[2]{
  \foreach \i in {1,...,#1}{\node[vtx] at (v\i) {\i};}
  \node[apex] at (va) {#2};
}

\newcommand{\PM}{\mathcal{PM}}

\newcommand{\sig}{\operatorname{sig}}

\begin{document}

\preprint{APS/123-QED}

\title{Generating Dicke State Graphs}

\author{Rebekah Herrman\orcidlink{0000-0001-6944-4206}}

\email{rherrma2@utk.edu}
\affiliation{Department of Industrial and Systems Engineering, University of Tennessee Knoxville, USA}
\date{\today}
\begin{abstract}
Graph theory is a powerful tool in quantum computing, with applications ranging from quantum circuit synthesis and optimization to entanglement mapping. Recent work has shown how one can use edge-colored graphs to model photonic experiments that generate GHZ and W states. However, the latter work also proved that verifying that a graph models a Dicke state experiment is coNP-complete. In this work, we provide families of graphs that generate $\ket{D_{k}^a} \otimes \ket{0}^{\otimes b }$, where $b = |a-2k|$ is the number of spectator modes. The
graph setup consists of a doubled complete subgraph on $a$ vertices and a collection of auxiliary vertices. We prove that every coincidence carries exactly $k$ excitations, every weight-$k$ computational basis state on bitstrings of length $a$ is realized, and each of those bitstrings is realized exactly $(n/2)!$ times, where $n = a+b$. Since verifying the Dicke FORALL condition is coNP-complete in general, constructing explicit families that provably generate Dicke states is of interest. 
\end{abstract}

\maketitle

\section{Introduction}\label{sec:intro}
Dicke states have broad applications in quantum algorithm design and quantum information science. For example, one may want to solve combinatorial optimization problems that require selecting exactly $k$ elements in a space of $n$ items using the Quantum Approximate Optimization Algorithm (QAOA) or its variants \cite{farhi2014quantum}. This would require generating Dicke states as initial input for the algorithm or using Dicke state preparation routines to create custom mixers, as in Grover-Mixer QAOA \cite{bartschi2020grover} or the Quantum Alternating Operator Ansatz \cite{hadfield2019quantum}. In a similar vein, bitstrings of Hamming weight $k$ can be used to generate constraints in the Variational Quantum Eigensolver \cite{wang2025variational, marti2025spin}. Other applications that use Dicke states include quantum metrology, error correction, and communication \cite{saleem2024achieving, aydin2026quantum, prevedel2009experimental, illiano2023quantum, zou2018beating, aktar2023scalable}. There exist several efficient quantum circuits that can generate Dicke and near-Dicke states, and recent experimental work has benchmarked Dicke state preparation on gate-based hardware \cite{bartschi2022short, aktar2022divide, lemr2009conditional, wang2021preparing, stojanovic2023dicke, nepomechie2023qudit, vittal2025efficient}. 

However, generating Dicke states in other models of quantum computing is more challenging. Recent work has shown that one can represent photonic quantum experiments using edge-colored weighted graphs, $G= (V,E)$ \cite{KrennGuZeilinger2017, GuGraphsII2019, GuGraphsIII2019, KrennPathIdentity2017}. In this model, a vertex of a graph represents an optical path, edges represent crystals, and edge colors represent the mode emitted into each arm. This work showed that perfect matchings of the experiment graph correspond to an $n$-fold coincidence, where each matching contributes a basis state determined by the coloring of the edges in the matching. Graphs that generate high dimensional GHZ states have been extensively studied in this context \cite{ChandranGajjalaIllickan2024, ChandranGajjala2024}. Vardi and Zhang formalize perfect matching under vertex-color constraints (PMVC) in \cite{VardiZhang2022, VardiZhang2023}. The EXISTS-PVMC condition asks if there exists a perfect matching under vertex-color
constraints in a graph with bi-colored edges, while FORALL-PVMC asks if, for a bi-colored graph $G$ and a set of $C \subseteq \{1, \ldots, d\}^d$ allowed vertex colorings, if for each $c \in C$, does $G$ have at least one perfect matching with inhered vertex coloring $c$.
 Dicke The EXISTS condition is polynomial-time checkable, whereas Gajjala, Ray and Thilikos proved that Dicke FORALL-PMVC is coNP-complete \cite{gajjala2026w}, resolving an open problem of \cite{VardiZhang2023}. 

In this work, we prove that there exists a family of graphs that can generate $\ket{D_k^a}$ or $\ket{D_k^a} \otimes \ket{0}^{\otimes b}$ where $b$ is a collection of spectator modes. The approach relies on complete graphs on $a$ vertices, denoted $K_a$, that have double edges, all of weight one, between all pairs of vertices. We prove that every coincidence carries exactly $k$ excitations, every weight-$k$ computational basis state on the signal modes is realized, and each of the $\binom{a}{k}$ bitstrings of Hamming weight $k$ is realized $(n/2)!$ times where $n = a+b$. In particular, Theorem~\ref{thm:main} establishes the FORALL condition of Dicke FORALL-PMVC for a family of graphs and Theorem~\ref{thm:count} shows that for any two inherited colorings $c, c'$, both $c$ and $c'$ are generated by the same number of matchings, implying each basis state has equal amplitude.

This work is organized as follows. In Section~\ref{sec:background}, we introduce graph theory and photonic quantum experiment terms. Then, in Section~\ref{sec:gadget}, we define the doubled complete core gadget that will be used to generate arbitrary states of the form $\ket{D_k^a} \otimes \ket{0}^{\otimes b}$ in Section~\ref{sec:algorithm}. Finally, we conclude with a discussion in Section~\ref{sec:discussion}. 

\section{Background}\label{sec:background}
In this section, we define useful graph theory and quantum information terminology.

\subsection{Graph theory}
In this work, we consider multigraphs that do not contain self-loops. A graph $G = (V(G),E(G))$ is a collection of points called \textit{vertices} and lines that connect them, called \textit{edges}. An equivalent definition is that an edge contains endpoints $i$ and $j$, for $i \neq j$, $i, j \in V(G)$. For a simple graph, there can exist at most one edge between any pair of vertices, however in a multigraph, there can be more than one edge between a given pair of vertices. If there exists more than one edge between vertices $i$ and $j$, we say that these edges are \textit{parallel}. The key gadget in this work consists of a complete graph on $a$ vertices, denoted $K_a$, with two copies of each edge.

An \textit{edge coloring} of $G$ is a function $f: E \rightarrow C$ where $C$ is a set of colors. In this work, we define $C= \{ \texttt{red}, \texttt{blue} \}$. A \textit{half-edge coloring} is similarly a function $c$ that assigns colors to each half of every edge in $G$. For example, if $G$ contains the edge $ij$, one can assign $\texttt{red}$ or $\texttt{blue}$ to the half edge incident to $i$ and also to the half edge incident to $j$, as in Fig.~\ref{fig:k6-full}. Note that this may result in monochromatic edges if the colors assigned to each half edge are the same. 

A \textit{matching} $M \subseteq E$ of a graph is a collection of edges where no two edges are incident to a common vertex. A matching is \textit{perfect} if it covers every $v \in V$. We denote the set of perfect matchings of a graph as $\mathcal{PM}(G)$. As $G$ will be an experiment graph, defined below, we omit the $(G)$ when it is obvious. Fig.~\ref{fig:n6-matching} is an example of a perfect matching. Notably, only graphs with an even number of vertices can have perfect matchings. 

Although not an explicit graph theory definition, here we define the \textit{double factorial} of a natural number $n$, $n!!$, as the product of all numbers between $1$ and $n$, inclusive, with the same parity as $n$ modulo 2. If $n$ is even, $n!! = 2 \times 4 \times ... \times n$, and if $n$ is odd, $n!! = 1 \times 3 \times ... \times n$. We will require the double factorial when counting the number of times an edge-colored graph can generate a particular bitstring. 

\subsection{Connecting graphs and photonic quantum experiments}
The focus of this work is \textit{Dicke states}. A Dicke state is defined as
\begin{equation*}
\ket{D_k^N} = \binom{N}{k}^{-1/2} \sum_{\substack{x \in \{0,1\}^N \\ H(x) = k}} \ket{x},  
\end{equation*}
where $H(x)$ is the \textit{Hamming weight} of bitstring $x=x_1x_2 \ldots x_{N}$, defined as $|\{x_i: x_i=1\}|$. We now translate the creation of Dicke states in photonic systems to a graph theory problem.
The experiment graph is based on the conventions in \cite{KrennGuZeilinger2017}. In this convention, a vertex $v \in V$ is an optical path and an edge $e\in E$ is a spontaneous parametric down conversion (SPDC) crystal. The mode emitted into arm $u$ of crystal $e$ is the half-edge color of edge $e$ (red emits $\ket{1}$ and blue emits $\ket{0}$), and the $n$-fold coincidence is a perfect matching of $G$. The basis state registered by a perfect matching $M$ is a bitstring $x= x_1x_2\ldots x_{n} \in \{0,1\}^n$ where $x_i = 0$ if $c(h_i(e_i(M)))=\texttt{blue}$ and  $x_i = 1$ if $c(h_i(e_i(M)))= \texttt{red}$. 
 We now define an \textit{experiment graph}.
\begin{definition}[Experiment graph]\label{def:eg}
An \emph{experiment graph} is a loopless multigraph $G=(V,E)$ together with a
half-edge coloring $c$ assigning to each arm of each crystal the emitted mode
$\mathrm{R}$ (excited, $\ket 1$) or $\mathrm{B}$ (ground, $\ket 0$). For a
crystal $e$ let $R(e)\subseteq e$ be the set of arms emitting $\mathrm{R}$, and
$\rho(e)=|R(e)|\in\{0,1,2\}$ its \emph{excitation yield}. The setup is
\emph{single-excitation} if $\rho(e)\le 1$ for every crystal, and
\emph{excitation-balanced} if $\rho(e)=1$ for every crystal.
\end{definition}
 Physically, $\rho(e)=1$ indicates that the crystal emits the pair $\ket{1,0}$ or
$\ket{0,1}$, while $\rho(e)=0$ indicates the crystal emits the pair $\ket{0,0}$. In this work, we do not use $\rho(e)=2$ (a $\ket{1,1}$ crystal) as in \cite{gajjala2026w}. In Section~\ref{sec:gadget}, we will define a base experiment graph, and use it to generate graphs that implement Dicke states. We now define the registered basis states of an experiment graph.

\begin{definition}[Registered basis state]\label{def:sig}
Let $M\in\mathcal{PM}(G)$ be a coincidence. Each path $i$ receives exactly one photon
from the unique crystal $e_i(M)\in M$ feeding it. Define
$\sig(M)\in\{0,1\}^{n}$ such that $\sig(M)_i=1$ if and only if photon $i$ is excited. Writing
bitstrings as subsets, $\sig(M)=\bigsqcup_{e\in M}R(e)$. With all crystal
weights equal to $1$, the post-selected state is
\[
  \ket{\psi_G}\;\propto\;\sum_{M\in\PM(G)}\ket{\sig(M)}
  \;=\;\sum_{x\in\{0,1\}^n}\mu(x)\,\ket{x},
  \qquad \mu(x)=\bigl|\sig^{-1}(x)\bigr|.
\]
\end{definition}
Throughout this work, we want to count the Hamming weight of $\sig(M)$ to prove that all basis states of Hamming weight $k$ appear the same number of times when we enumerate over the registered basis states of all perfect matchings.
\begin{lemma}[Excitation accounting]\label{lem:additive}
For every coincidence $M$, the number of excited detectors is
$H(\sig(M))=\sum_{e\in M}\rho(e)$.
\end{lemma}

\begin{proof}
Distinct crystals of a coincidence feed disjoint paths, so the sets $R(e)$, where
$e\in M$, are pairwise disjoint and their union is $\sig(M)$.
\end{proof}

The next result bounds the Hamming weight $H(\sig(M))$. 
\begin{proposition}[Mode-number parity and the excitation ceiling]\label{prop:ceiling}
Let $(G,c)$ be a single-excitation experiment graph with coloring $c$ admitting at least one coincidence.
Then
\begin{enumerate}
\item[(i)] the number of paths $n$ is even and
\item[(ii)] every coincidence registers at most $n/2$ excitations,
$H(\sig(M))\le n/2$, with equality for all $M$ if and only if every crystal used in some
coincidence is excitation-balanced.
\end{enumerate}
Consequently no such setup produces a state on an odd number of modes, and none
produces excitation number $k$ with fewer than $2k$ modes.
\end{proposition}

\begin{proof}
(i) A coincidence pairs up all paths. (ii) A coincidence consists of $n/2$
crystals, each contributing at most one excitation by hypothesis. Applying
Lemma~\ref{lem:additive} results in $H(\sig(M))\le n/2$. Conversely, equality for all $M$ forces $\rho(e) =1$ on every crystal lying in some coincidence, as a single $\rho(e)=0$ for some $e \in M$ results in $H(\sig(M)) < n/2$.
\end{proof}

It is worth emphasizing that each crystal can inject at most one excitation into a coincidence, so
$k$ excitations require at least $k$ crystals, hence $n\ge 2k$ paths.

\section{Hamming weight gadgets}\label{sec:gadget}

Suppose we want to generate the Dicke state $\ket{D_d^{a}}$ where $a=2d$ for some $d \in \mathbb{N}$. Let us consider an $a$-vertex multigraph. Let the edges contain two copies of $ij$ for all $i \neq j, i,j \in [a]$. One copy of edge $ij$ will be half colored so that the red half-edge is incident to $i$ and blue half-edge is incident to $j$. The other copy of the edge $ij$ has the reverse coloring, so the blue half-edge is incident to $i$ and the red half-edge is incident to $j$, as in Fig.~\ref{fig:k6-full}.

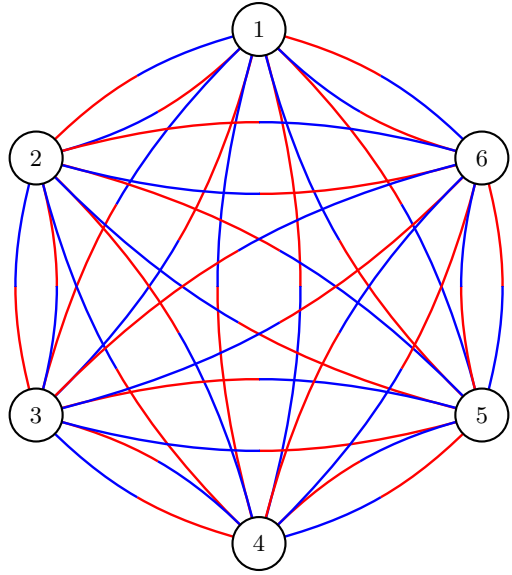
\begin{figure}
  \centering
  \begin{tikzpicture}
    \placevertices{6}{3.4}{}

    \foreach \i in {1,...,5}{
      \foreach \j in {1,...,6}{
        \ifnum\i<\j
          \halfedge{v\i}{v\j}{16}   
          \halfedge{v\j}{v\i}{16}   
        \fi
      }
    }

    \foreach \i in {1,...,6}{\node[vtx] at (v\i) {\i};}
  \end{tikzpicture}
  \caption{The doubled complete graph on six vertices. Every pair
    $\{i,j\}$ carries two parallel edges: one red at $i$ and blue at $j$,
    the other red at $j$ and blue at $i$, for $30$ edges in total. Since
    each edge contributes exactly one red half and a perfect matching has
    three edges, every matching yields a bitstring of Hamming weight
    $k=3$.}
  \label{fig:k6-full}
\end{figure}

First, note that graphs of this form have $2^{d}((2d-1)!!)$ perfect matchings, as complete graphs on $a=2d$ vertices have $(2d-1)!!$ perfect matchings. The $2^{d}$ factor comes from the fact that each matching has $\frac{n}{2} = d$ edges, and for each pair of vertices, $i,j$, there are two choices of edge $ij$ that can be included in the matching. Clearly, each bitstring that results from this coloring has Hamming weight $d$, as every edge has one red half-edge. We claim that 1) this graph generates all bitstrings of Hamming weight $d$, and 2) when enumerating the set of bitstrings that correspond to every single perfect matching of this graph, each bitstring appears in the set the same number of times. However, we shall prove this in generality in Theorems~\ref{thm:main} and ~\ref{thm:count} below, so we do not prove it for this specialized case here. 

 Notably, since $n=a+b$, both $a$ and $b$ must be either even or odd in order for $n$ to be even, so $a \equiv b \pmod 2$. Thus, if we want to generate $\ket{D_k^{2d-1}}$, we require the use of an additional vertex, and instead generate $\ket{D_k^{2d-1}} \otimes \ket{0}$ by enforcing that all half-edges incident to vertex $2d$ are $\texttt{blue}$. To do this, we use a copy of $K_{2d-1}$ with two sets of edges as above. We then add one more vertex, $2d$, and connect it to every vertex between $1$ and $2d-1$ with a monochromatic blue edge, which we denote as the 
\textit{sink construction}. See Fig.~\ref{fig:n6-full} for the graph that generates $\ket{D_2^5} \otimes \ket{0}$. Notably this graph generates all bitstrings $x$ that have Hamming weight $\lfloor \frac{2d-1}{2}\rfloor$, where the last digit of the bitstring, $x_{2d}$, is always $0$, thus we can discard it. Again, we shall prove the general result in the next section, so we do not prove this specific case here. Similarly, if we connect all vertices in $K_5$ using a red half edge at $i$ for all $i \in V(K_5)$ and a blue half-edge at $6$, we can generate $\ket{D_3^5} \otimes \ket{0}$. We call this approach the \textit{source construction}. We now formally define these gadgets.

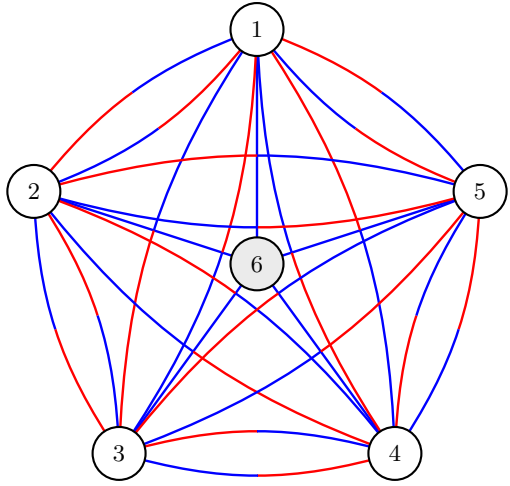
\begin{figure}
  \centering
  \begin{tikzpicture}[scale=1]
    \placevertices{5}{3.1}{6}

    \foreach \i in {1,...,5}{\draw[spoke] (v\i) -- (va);}

    \foreach \i in {1,...,4}{
      \foreach \j in {1,...,5}{
        \ifnum\i<\j
          \halfedge{v\i}{v\j}{16}   
          \halfedge{v\j}{v\i}{16}   
        \fi
      }
    }

    \drawvertices{5}{6}
  \end{tikzpicture}
  \caption{The graph that generates $\ket{D_2^5} \otimes \ket{0}$. It consists of a doubled $K_5$ on vertices $1,\dots,5$ in
    which every pair $\{i,j\}$ carries two parallel edges, each half red and half blue, such that one edge has red
    half at $i$ and one edge has red half at $j$, together with all-blue
    spokes joining each vertex to vertex $6$. This is an example of a \textit{sink construction}. If the half edges from $K_5$ to $6$ were all colored red at $i \in K_5$ and blue at $6$, it would generate $\ket{D_3^5} \otimes \ket{0}$ using a \textit{source construction}.}
  \label{fig:n6-full}
\end{figure}

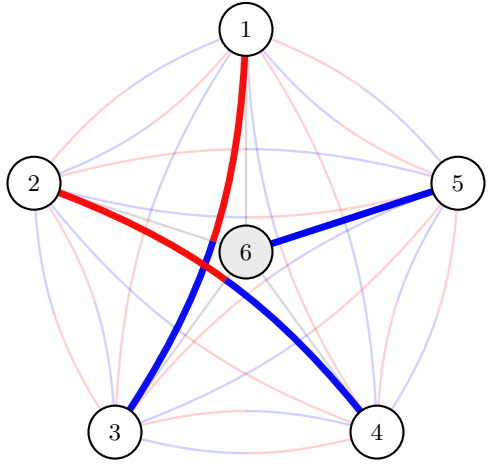
\begin{figure}
  \centering
  \begin{tikzpicture}[scale=0.95]
    \placevertices{5}{3.1}{6}

    \foreach \i in {1,...,5}{\draw[spoke,black!15] (v\i) -- (va);}
    \foreach \i in {1,...,4}{
      \foreach \j in {1,...,5}{
        \ifnum\i<\j
          \begin{scope}[opacity=0.18]
            \halfedge{v\i}{v\j}{16}
            \halfedge{v\j}{v\i}{16}
          \end{scope}
        \fi
      }
    }

    \halfedgehi{v1}{v3}{16}       
    \halfedgehi{v2}{v4}{16}       
    \draw[spokehi] (v5) -- (va);  

    \drawvertices{5}{6}
  \end{tikzpicture}
  \caption{A perfect matching $\{13,24,56\}$ in the graph of
    Fig.~\ref{fig:n6-full}, with red halves at $1$ and $2$, yielding the
    bitstring $110000$. Each of the $\binom{5}{2}=10$ weight-two bitstrings
    arises from exactly $3!=6$ of the $60$ perfect matchings.}
  \label{fig:n6-matching}
\end{figure}

\begin{definition}[Doubled complete core gadget]\label{def:doubled}
Let $A$ be a set of $a$ signal paths. The \emph{doubled complete core gadget}
on $A$ consists of $a(a-1)$ crystals. For each pair of vertices of $A$ $\{i,j\}\subseteq A$, there are two crystals. One
emits $\ket{1_i,0_j}$ and one emits $\ket{0_i,1_j}$. 
\end{definition}
We now define $\texttt{G}(a, b, \sigma)$. 
\begin{definition}[$\texttt{G}(a,b,\sigma)$]\label{def:gadget}
Let $a\ge 1$, $b\ge 0$ with $b\le a$ and $a\equiv b\pmod 2$, and let
$\sigma\in\{\text{sink},\text{source}\}$ (with $\sigma=\text{sink}$ when $b=0$).
Take the doubled complete core gadget on $A$, add a set $B$ of $b$ \emph{spectator paths}
carrying \emph{no} crystals among themselves, and join each $v\in B$ to each
$i\in A$ by a single crystal $\sigma$, which is one of
\[
  \text{sink:}\ \ \ket{0_i,0_v},\qquad\qquad
  \text{source:}\ \ \ket{1_i,0_v}.
\]
Thus $n=a+b$ paths and $a(a-1)+ab$ crystals. Set
\[
  m=\frac{a-b}{2},\qquad
  k=\begin{cases} m=\dfrac{a-b}{2}, & \sigma=\text{sink},\\[2mm]
                  \dfrac{a+b}{2}, & \sigma=\text{source}.\end{cases}
\]
\end{definition}

A spectator path always
absorbs one signal path into a coincidence. In a sink construction that crystal injects
no excitation, so it costs one. In the source construction it injects one into the
signal path, so it pays one. Every spectator therefore moves $k$ up or down one unit. We now prove one of the key lemmas needed to prove that these gadgets can generate Dicke states.

\begin{lemma}[Canonical form of a coincidence]\label{lem:canon}
Every coincidence of $\texttt{G}(a,b,\sigma)$ arises uniquely as an injection
$\phi:B\hookrightarrow A$ (each spectator takes a distinct signal path), together
with a coincidence of the doubled complete core on the remaining $a-b=2m$ signal paths.
Conversely every such pair is a coincidence.
\end{lemma}

\begin{proof}
Since $B$ carries no internal crystals (i.e. there are no edges between any two vertices of $B$) for any coincidence, every spectator is fed by a crystal reaching into $A$.  Since $M$ is a matching, distinct spectators use distinct signal paths,
giving an injection $\phi$. The signal paths outside $\phi(B)$ are fed only by
core crystals and must be paired among themselves, which is exactly a
coincidence of the doubled complete core gadget they span. Reversing these arguments proves the converse.
\end{proof}

\begin{theorem}\label{thm:weight}
Every coincidence of $\texttt{G}(a,b,\sigma)$ registers exactly $k$ excitations, all
of them on signal paths and every spectator path registers $\ket 0$ in every
coincidence. Equivalently, $\ket{\psi_{\texttt{G}}}$ is supported on the weight-$k$
sector of the signal modes tensored with vacuum on $B$.
\end{theorem}

\begin{proof}
Decompose a coincidence as in Lemma~\ref{lem:canon}. The $m$ core crystals are
excitation-balanced and contribute $m$ excitations, all on signal paths. The
$b$ spectator crystals contribute $0$ each in a sink construction and $1$ each (on side $A$) in a source construction. Lemma~\ref{lem:additive} gives
weight $m$ or $m+b$, i.e.\ $(a-b)/2$ or $(a+b)/2$. Every crystal touching a
spectator emits $\ket 0$ into it, so $\sig(M)$ vanishes on $B$.
\end{proof} 
\noindent These results imply that $\texttt{G}(a,b,\sigma)$ satisfies Dicke \textsc{EXISTS-PMVC} by
construction.

\section{General Dicke State Graphs}\label{sec:algorithm}
In order to construct the state $\ket{D_k^a}$, we need to develop a graph that consists of at least $|A|=|V(K_a)|=a$ vertices. Let $B$ be the set of spectator modes and let $|B|=b$.  We claim that we can use the source and sink constructions above to generate  states of the form $\ket{D_{k}^a} \otimes \ket{0}^{\otimes b}$, where $b=|a-2k|$. Note that we can achieve bitstrings of any Hamming weight between $0$ and $a$ by varying the size of $B$ and using either the source or sink construction. We shall always require that $n=a+b$ is even, as only even graphs have perfect matchings. This implies that the cardinality of both $A$ and $B$ must be odd or must both be even, or $a \equiv b \pmod 2$.

In order to prove that these graphs generate $\ket{D_k^a }\otimes \ket{0}^{\otimes b}$, we need to prove that every bitstring of length $a$ and Hamming weight $k$ can be generated by a perfect matching of $\texttt{G}\;(a, b, \sigma)$ and that, in the set of perfect matchings of $\texttt{G}\;(a, b, \sigma)$, every bitstring of length $a$ and Hamming weight $k$ is generated the same number of times. Formally the former condition is

\begin{theorem}\label{thm:main}
    Fix $\texttt{G}\;(a, b, \sigma)$ and let $n=a+b$. For every $S \subseteq A$ such that $|S| = k$, the number of coincidences registering the excitation pattern $S$ is $\mu (S) = \large( \frac{n}{2} \large) !$. In particular, all $\binom{a}{k}$ patterns of weight $k$ on the first $a$ modes occur and 
    \begin{equation*}
        \ket{\psi_{\texttt{G}(a, b, \sigma)}} \propto\ket{D_k^a} \otimes \ket{0}^{\otimes b}.
    \end{equation*}
\end{theorem}

\begin{proof}
Define $T=A\setminus S$. We need prove the result when using both source and sink constructions.

\noindent \emph{Sink construction:} ($k=m$, $|T|=a-k=\tfrac{a+b}{2}=\tfrac n2$). By
Lemma~\ref{lem:canon} a coincidence registering $S$ consists of a choice of
which $b$ elements of $T$ are taken by spectators together with the injection
realizing it, and a pairing of $S$ with the remaining $k$ elements of $T$
(each such core crystal must be the copy emitting $\ket 1$ into the $S$-side,
which is forced). No element of $S$ may be taken by a spectator, since a sink
crystal would leave it unexcited. Hence
\[
  \mu(S)=\underbrace{\frac{(a-k)!}{(a-k-b)!}}_{\text{injections }B\hookrightarrow T}
         \cdot\underbrace{k!}_{\text{bijections }S\to T\setminus\phi(B)}
        =\frac{(a-k)!}{k!}\cdot k!=(a-k)!=\Bigl(\tfrac n2\Bigr)!,
\]
using $a-k-b=k$.

\noindent \emph{Source construction:} ($k=\tfrac{a+b}{2}$, $|T|=a-k=m=\tfrac{a-b}{2}$, and
$\tfrac n2=k$). Now every spectator excites the signal path it takes, so
$\phi(B)\subseteq S$. The remaining $k-b$ elements of $S$ pair bijectively with
the $m=k-b$ elements of $T$, again with the emitting copy forced. Hence
\[
  \mu(S)=\frac{k!}{(k-b)!}\cdot (k-b)! = k! = \Bigl(\tfrac n2\Bigr)!.
\]
In both cases the count is independent of $S$ and nonzero, so every weight-$k$
pattern on $A$ is realized with the same amplitude, and
$\ket{\psi}\propto\sum_{|S|=k}\ket{S}\otimes\ket{0^b}$.
\end{proof}

Formally, the latter condition is 
\begin{theorem}[Coincidence count]\label{thm:count}
The number of coincidences of $\texttt{G}(a,b,\sigma)$ is
\[
  M(a,b)\;=\;\frac{a!}{\bigl(\tfrac{a-b}{2}\bigr)!}\;=\;\binom{a}{k}\Bigl(\tfrac n2\Bigr)!,
\]
regardless of choice of $\sigma$.
\end{theorem}

\begin{proof}
Since Lemma~\ref{lem:canon} shows that each coincidence of $\texttt{G}(a,b,\sigma)$ arises uniquely as an injection from $B$ into $A$ and a doubled complete core on $a-b = 2m$ signal paths, we can multiply the number of injections and the number of doubled complete core signals to count $M(a,b)$. Note that an injection from $B$ into $A$ is an ordered $b$-tuple of distinct elements of $A$, which is counted as $\frac{a!}{(a-b)!}$. As there remain $2m=a-b$ signal paths on a double completed core gadget, there are $(2m-1)!!$ perfect matchings. Each of the $m$ pairs in the doubled complete core can be realized by two edges, hence a factor of $2^m$. Thus, $M(a,b)=\frac{a!}{(a-b)!}\cdot(2m-1)!!\cdot 2^{m}$. Using
$(2m-1)!!\,2^{m}=(2m)!/m!$ and $2m=a-b$ yields $M(a,b)=a!/m!$. The second expression
is Theorem~\ref{thm:main} combined with completeness.
\end{proof}

Note that each spectator--signal pair carries one
crystal in either the source or sink construction, so the two setups have the same set
of coincidences and differ only in which photons are excited. Changing the
emitted modes of the $ab$ spectator crystals moves the target state from
$\ket{D_{(a-b)/2}^a}$ to $\ket{D_{(a+b)/2}^a}$ at zero change in crystal count,
coincidence count, or detection rate.

\begin{corollary}[Design rule]\label{cor:design}
For any $a\ge 1$ and $0\le k\le a$ there is a setup of the above type producing
$\ket{D_k^a}$ on $a$ signal modes, using
$n=2\max(k,\,a-k)\ \text{ paths},\; b=|a-2k|\ \text{spectators},$ and $
  a(a-1)+ab\ \text{crystals}$
with $\sigma=\mathrm{sink}$ if $k< a/2$ and $\sigma=\mathrm{source}$ if
$k > a/2$. If $k = a/2$, we need no spectator modes. The amplitude of each term is $(n/2)!$ and the number of
contributing coincidences is $a!/\bigl(\min(k,a-k)\bigr)!$.
\end{corollary}

\begin{proof}
Solve $k=(a\mp b)/2$ for $b=|a-2k|\ge 0$; then $a\equiv b \pmod 2$ holds
automatically and $n=a+b=2\max(k,a-k)$. The result follows from applying Theorems~\ref{thm:weight},
\ref{thm:main} and~\ref{thm:count}.
\end{proof}

Notably, this construction is mode-optimal for $k \geq a/2$.
\begin{corollary}[Optimality above half filling]\label{cor:opt}
If $k\ge a/2$ then $n=2k$, which meets the lower bound of
Proposition~\ref{prop:ceiling}(ii). No single-excitation setup
produces excitation number $k$ on fewer than $2k$ paths, so the source
configuration is mode-optimal in that regime. For $k<a/2$ the construction uses
$n=2(a-k)$ paths, which exceeds the bound $n\ge 2k$.
\end{corollary}
In particular, the double core gadget construction may not be mode-optimal when $k < a/2$. This is especially obvious in the W-state case, as the construction in \cite{gajjala2026w, GuGraphsIII2019} creates W states using $a$ paths with no spectators.

We conclude with an example and observation that vertex-transitivity is not sufficient to generate Dicke states. To see this, consider the double complete bipartite graph $K_{3,3}$ which consists of two sets of three vertices, with doubled edges, representing crystals $\ket{1,0}$ and $\ket{0,1}$, occurring only between two vertices not in the same sets, as in Fig.~\ref{fig:k33-full}.  This experiment graph is
vertex-transitive and edge-transitive, is excitation-balanced, and every
coincidence carries $3$ excitations. All $20$ weight-$3$ patterns occur, however the bitstrings $111000$ and $000111$ occur $3!=6$ times, where the first three bits refer to vertices in one partition and the last three bits refer to vertices in the other partition. All other bitstrings occur only two times when enumerating over all perfect matchings of the experiment graph.  The resulting state is a superposition of weight-$3$
basis states but is not $\ket{D_3^6}$.
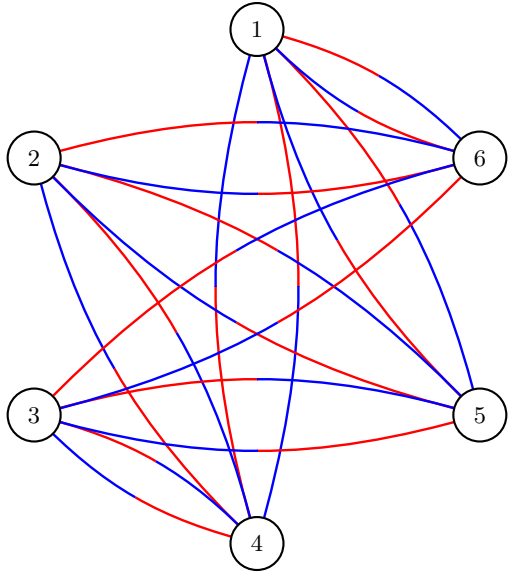
\begin{figure}
  \centering
  \begin{tikzpicture}
    \placevertices{6}{3.4}{}

    \foreach \i in {1,2,3}{
      \foreach \j in {4,5,6}{
        \halfedge{v\i}{v\j}{16}   
        \halfedge{v\j}{v\i}{16}   
      }
    }

    \foreach \i in {1,...,6}{\node[vtx] at (v\i) {\i};}
  \end{tikzpicture}
  \caption{The doubled complete bipartite experiment graph on
    $X=\{1,2,3\}$ and $Y=\{4,5,6\}$. Every pair $\{i,j\}$ with $i\in X$,
    $j\in Y$ carries two parallel crystals, one emitting
    $\lvert 1_i,0_j\rangle$ and one emitting $\lvert 0_i,1_j\rangle$,
    for $18$ crystals in total. Every crystal is excitation-balanced and every
    coincidence uses three crystals, so every coincidence registers
    Hamming weight $3$.}
  \label{fig:k33-full}
\end{figure}

\section{Conclusion}\label{sec:discussion}
In this work, for every $a$ and $k$, we show how to construct a quantum state of the form $\ket{D_k^a} \otimes \ket{0}^{\otimes b}$, where $b = |a-2k|$. When $a$ is odd, we note that one needs a spectator mode to construct $\ket{D_k^a}$, as perfect matchings exist only on graphs that have an even number of vertices. Additionally, we prove that when $k \geq \frac{a}{2}$, the construction is mode-optimal. It is worth noting that the W states are the case when $k=1$ and that the construction in this work is not optimal for generating W states. Furthermore, we give an example where a vertex- and edge-transitive graph does not generate a Dicke state, highlighting that more rigid restrictions are required to generate Dicke states when using a doubled complete core gadget. Throughout this work, all graphs have an implicit edge weight of $1$ and all amplitudes are positive integers, so this family of graphs lies in the constructive-interference regime.

We note that this graph can generate $\ket{D_k^a} \otimes \ket{0}^{\otimes b}$ with  $a(a-1)+ab$ crystals, so modest $\ket{D_k^a}$ should be implementable on hardware \cite{bao2023very}. Future research includes developing graphs that can generate Dicke states with fewer edges. Furthermore, while this work generates Dicke states such that each basis state has the same amplitude, we showed that one can create asymmetric Dicke states, i.e. states that have the same basis states as Dicke states but with unequal amplitudes in the example at the end of the previous section. A mathematically rigorous characterization of the types of states that can be generated by applying graph operations, such as adding and deleting edges, to the source and sink gadgets would yield a larger collection of quantum states that can be generated using graphs. Finally, as mentioned in \cite{gajjala2026w}, structural characterizations of graphs that can generate Dicke states is a natural open question.

\section*{AI use}
The author used Claude Opus 5 to verify the FORALL Dicke state condition for sample graphs, write preliminary drafts of the mathematical statements, and to generate Tikz code. The author takes sole responsibility for correctness of the final manuscript.

\section*{Competing interests}
All authors declare no financial or non-financial competing interests.

\bibliographystyle{ieeetr}
\bibliography{main}
\end{document}